\pdfoutput=1
\documentclass[letterpaper, 10 pt, conference]{ieeeconf}  % Comment this line out if you need a4paper

\IEEEoverridecommandlockouts                              % This command is only needed if 
\let\labelindent\relax

\usepackage{mathptmx}       % selects Times Roman as basic font
\usepackage{helvet}         % selects Helvetica as sans-serif font
\usepackage{courier}        % selects Courier as typewriter font
\usepackage{type1cm}        % activate if the above 3 fonts are
   \usepackage{lipsum}                         % not available on your system
\usepackage{makeidx}         % allows index generation
\usepackage{graphicx}        % standard LaTeX graphics tool
\usepackage{multicol}        % used for the two-column index
\usepackage[bottom]{footmisc}% places footnotes at page bottom
\usepackage{comment}
\usepackage{soul}
\RequirePackage{ifpdf}
\usepackage{amsmath,amssymb,amsbsy}

\usepackage{amsthm} % assumes new font selection scheme installed
\newtheorem{theorem}{Theorem}[section]
\newtheorem{lemma}[theorem]{Lemma}
\newtheorem{definition}[theorem]{Definition}	
\newtheorem{proposition}[theorem]{Proposition}
	
\newtheorem{remark}[theorem]{Remark}

\usepackage[T1]{fontenc}      % Use Type-1 fonts

\usepackage{mathptmx}         % selects Times Roman as basic font
\usepackage{helvet}           % selects Helvetica as sans-serif font
\usepackage{courier}          % selects Courier as typewriter font
\usepackage{type1cm}          % activate if the above 3 fonts are
\usepackage{calc}
\usepackage[shortlabels]{enumitem}
\usepackage[shortcuts]{extdash}
\usepackage{graphicx}
\usepackage{capt-of}
\usepackage{booktabs}
\usepackage{psfrag}
\usepackage{epstopdf}
\ifpdf\DeclareGraphicsExtensions{.jpg,.png,.pdf}\fi

\usepackage[bottom]{footmisc} % places footnotes at page bottom

\makeatletter
\def\citet{\@ifstar{\citetstar}{\citetnostar}}
\def\Citet{\@ifstar{\Citetstar}{\Citetnostar}}
\def\citetnostar{\@ifnextchar[{\squarecitet}{\simplecitet}}
\def\squarecitet[#1]{\@ifnextchar[{\twocitet[#1]}{\onecitet[#1]}}
\def\Citetnostar{\@ifnextchar[{\squareCitet}{\simpleCitet}}
\def\squareCitet[#1]{\@ifnextchar[{\twoCitet[#1]}{\oneCitet[#1]}}
\def\citetstar{\@ifnextchar[{\squarecitetstar}{\simplecitetstar}}
\def\squarecitetstar[#1]{\@ifnextchar[{\twocitetstar[#1]}{\onecitetstar[#1]}}
\def\Citetstar{\@ifnextchar[{\squareCitetstar}{\simpleCitetstar}}
\def\squareCitetstar[#1]{\@ifnextchar[{\twoCitetstar[#1]}{\oneCitetstar[#1]}}
\makeatother
\def\simplecitet#1{\citeauthor{#1}~\citeyearpar{#1}}
\def\onecitet[#1]#2{\citeauthor{#2}~\citeyearpar[#1]{#2}}
\def\twocitet[#1][#2]#3{\citeauthor{#3}~\citeyearpar[#1][#2]{#3}}
\def\simplecitetstar#1{\citeauthor*{#1}~\citeyearpar{#1}}
\def\onecitetstar[#1]#2{\citeauthor*{#2}~\citeyearpar[#1]{#2}}
\def\twocitetstar[#1][#2]#3{\citeauthor*{#3}~\citeyearpar[#1][#2]{#3}}
\def\simpleCitet#1{\Citeauthor{#1}~\citeyearpar{#1}}
\def\oneCitet[#1]#2{\Citeauthor{#2}~\citeyearpar[#1]{#2}}
\def\twoCitet[#1][#2]#3{\Citeauthor{#3}~\citeyearpar[#1][#2]{#3}}
\def\simpleCitetstar#1{\Citeauthor*{#1}~\citeyearpar{#1}}
\def\oneCitetstar[#1]#2{\Citeauthor*{#2}~\citeyearpar[#1]{#2}}
\def\twoCitetstar[#1][#2]#3{\Citeauthor*{#3}~\citeyearpar[#1][#2]{#3}}

\usepackage[cmex10]{mathtools}             % Loads and extends amsmath
\usepackage{amsfonts,amssymb}
\usepackage{mathrsfs}                      % Gives us \mathscr for \setset

\input{_preamble_hyperref}

\title{\LARGE \bf
Confinement Control of Double Integrators using Partially Periodic Leader Trajectories
}

\author{Karthik Elamvazhuthi, Sean Wilson,  and  Spring Berman% <-this % stops a space
\thanks{*This research was supported by NSF Award CMMI-1436960 and by DARPA Young Faculty Award D14AP00054.}% <-this % stops a space
\thanks{The authors are with the School for Engineering of Matter, Transport and Energy, Arizona State University, Tempe, AZ 85287, USA.  E-mail: {\tt\small \{karthikevaz, sean.t.wilson, spring.berman\}@asu.edu}}%
}

\begin{document}

\maketitle
\thispagestyle{empty}
\pagestyle{empty}

%%%%%%%%%%%%%%%%%%%%%%%%%%%%%%%%%%%%%%%%%%%%%%%%%%%%%%%%%%%%%%%%%%%%%%%%%%%%%%%%
\begin{abstract}
We consider a multi-agent confinement control problem in which a single leader has a purely repulsive effect on follower agents with double-integrator dynamics.  By decomposing the leader's control inputs into periodic and aperiodic components, we show that the leader can be driven so as to guarantee confinement of the followers about a time-dependent trajectory in the plane. We use tools from averaging theory and an input-to-state stability type argument to derive conditions on the model parameters that guarantee confinement of the followers about the trajectory. For the case of a single follower, we show that if the follower starts at the origin, then the error in trajectory tracking can be made arbitrarily small depending on the frequency of the periodic control components and the rate of change of the trajectory. We validate our approach using simulations and experiments with a small mobile robot.
\end{abstract}

%%%%%%%%%%%%%%%%%%%%%%%%%%%%%%%%%%%%%%%%%%%%%%%%%%%%%%%%%%%%%%%%%%%%%%%%%%%%%%%%
\section{Introduction}
In this work, we present a leader-follower control strategy for multi-robot applications such as exploration, environmental monitoring, disaster response, and targeted drug delivery at the micro-nanoscale.  We consider a scenario in which a leader agent must confine a group of follower agents to a certain region around a target trajectory while steering them along this trajectory.  For example, a leader agent with sophisticated sensing, localization, and planning capabilities may be required to herd a group of followers, which lack these capabilities, through an environment with obstacles.

%A simple motivation is scenarios in which the environment is populated with obstacles, but a majority of agents have lower capabilities in terms of planning and localization.

% targeted drug delivery at the micro-nanoscale

%biomedical applications at the micro-nanoscale such as imaging and targeted drug delivery.

The problem that we consider is closely related to containment control problems \cite{ren2010distributed} in the multi-agent systems literature. In these problems, the objective is to design interaction rules between leader and follower agents so that the followers are eventually contained in the convex hull spanned by the leaders.  In contrast to most other work on such problems, the system that we consider is naturally unstable for steady-state control inputs due to the repulsive nature of the interaction between the leader and the followers.  Hence, a stabilization mechanism is needed to achieve confinement of the followers. One way to stabilize the system is to incorporate feedback on the followers' positions and velocities into the leader trajectories.  Our control approach does not require this feedback, and instead uses {\it open-loop oscillatory strategies} to (practically) stabilize the follower positions about near-zero velocity conditions. %\kar{I am used to seeing most  controllability results being valid only near zero velocity conditions} \spr{should we emphasize some other aspect of the control strategy?}

% \kar{practical stability is a stability notion where a system condition is said to be practically stable if locally all solutions converge to a small neighborhood of the system condition: different from stable equilibrium point, where the system converges to a single point}

%However, we do not take this approach here. 

%On the other hand, we show that there exist open loop oscillatory strategies that (practically) stabilize the follower positions about near zero velocity conditions. 

The use of oscillatory inputs in control theory has a rich history. There has especially been extensive work on motion planning of driftless systems.  In this work, sinusoidal inputs of appropriately chosen frequencies are used to produce excitations along different independent directions (corresponding to Lie brackets of system vector fields) in the system's state space \cite{sussmann1991limits,murray1993nonholonomic,chitour2013global}. In fact, for a certain class of systems with a specific geometric structure, solutions of some optimal control problems are of a periodic nature \cite{brockett1982control,sastry1993structure}. Similarly, oscillatory inputs have also played an important role in stabilization problems \cite{baillieul1993stable,bellman1986vibrational,teel1995non,coron1992global,sarychev2001stability}. While most of these works have considered the control of systems with a range of rest conditions (equilibrium points), there has also been some work on practical stabilization of systems for which no such rest condition exists \cite{moreau2000practical}.
The application of oscillatory inputs can also be found in the physics literature. Kapitza's pendulum is an example of a system whose unstable equilibrium is rendered stable when subject to oscillatory vibrations \cite{kapitza1951dynamic}. The method of using oscillatory potentials for trapping charged particles is also well-known \cite{gilary2003trapping,paul1990electromagnetic}. 

Our approach is also motivated by animal foraging and livestock herding behaviors, which have inspired a number of algorithms and multi-robot control strategies \cite{haque2011biologically,piersonbio,bennett2012comparative,lien2004shepherding,colombo2012confinement,bressan2012control}.  The approaches in \cite{haque2011biologically,bressan2012control}, and \cite{colombo2012confinement} are most similar to ours in their use of sinusoidal leader trajectories that introduce repulsion terms into the agent dynamics.  While \cite{haque2011biologically} uses sinusoidal trajectories to mimic the foraging strategies of dolphins, the analysis is based on Snell's law and the agent interactions arise from prey ``bouncing'' off of predators by specular reflection.  In \cite{bressan2012control} and \cite{colombo2012confinement}, the follower population is modeled as a probability density over space and time, and the objective is to control the support of this density function.  The followers are diffusing at a finite speed, and their models are of single-integrator type. These two works derive conditions for confinement to be possible under these assumptions. Attraction-type potentials are also considered in \cite{colombo2012confinement}, and conditions are also derived for the confinement of followers in dimensions greater than two.  Although these works do not consider the trajectory tracking problem, \cite{bressan2012control} does study the existence of control inputs for point-to-point steering problems.   % (can be made arbitrarily large) 

%The individual dynamical models of the agents translate to that of single integrator dynamics 
%with additional repulsion terms based on leader locations. 

%In a robotic setting this is undesirable and might be a conservative estimate as we show further on (though it is clearly a much more challenging problem to control a diffusing flock).  

% , and the analysis is done in the more \spr{technically complex?} framework of differential inclusions.  

%An interesting result in \cite{colombo2012confinement} is that one can also achieve confinement in dimensions greater than or equal to 3 under certain conditions. It's not clear however what kind of control inputs would achieve this. 

Other works that consider similar problems use methods based on consensus protocols and multiple leaders \cite{ji2008containment,cao2012distributed}. Multi-agent controllability problems have also received some attention in related contexts \cite{caponigro2015sparse,rahmani2009controllability,tanner2004controllability,piccoli2014control,brockett2010control}. 

\section{Problem Statement}

We consider a system with a single {\it leader agent}, whose position at time $t$ is given by $[x_l(t),y_l(t)]^T$, and $N$ {\it follower agents}, whose positions at time $t$ are $[x^i(t),y^i(t)]^T$, $i=1,...,N$.  We denote the $x$ and $y$ velocity components of follower agent $i$ at time $t$ by $[v^i_x(t),v^i_y(t)]^T$.  The dynamics of the leader and followers are defined by the following system of equations:
\begin{eqnarray}
\label{eq:MainSys}
\dot{x}_l &=& u_x \\ \nonumber
\dot{y}_l &=& u_y \\ \nonumber
\dot{x}^i &=& v^i_x  \\ \nonumber
\dot{v}^i_x &=& \frac{A (x^i-x_l)}{[(x^i-x_l)^2+(y^i-y_l)^2]^{\alpha}} - k v^i_x\\ \nonumber
\dot{y}^i &=& v^i_y \\ \nonumber
\dot{v}^i_y &=& \frac{A (y^i-y_l)}{[(x^i-x_l)^2+(y^i-y_l)^2]^{\alpha}} - k v^i_y
\end{eqnarray}
where $A, k, \alpha \in \mathbb{R}^+$ and $i=1,...,N$.  This model produces one-way interaction between the leader and the followers for any non-zero value of $A$.  The interaction potentials are of gravitational type but repulsive in nature. % interaction

% For $\mathbf{x} \in \mathbb{R}^n$,
%\begin{equation}
%\|\mathbf{x}\| = \sqrt{x_1^2+x_2^2+....+x_n^2}. \nonumber
%\end{equation}

We formulate our confinement control problem using the following definition, in which $\|\mathbf{x}\|$ is the Euclidean norm of $\mathbf{x} \in \mathbb{R}^n$:
\begin{definition}
Let $\gamma: [0,1] \rightarrow \mathbb{R}^2$ and $R>0$. Then the system \ref{eq:MainSys} is said to be $R$-confinement controllable about the trajectory $\gamma$ if there exist $T>0$, $u_x:[0,T]\rightarrow \mathbb{R}$, and $u_y:[0,T]\rightarrow \mathbb{R}$ such that the solution of \ref{eq:MainSys} satisfies $\|\gamma(t/T)-[x^i(t),y^i(t)]^T\|<R$ and $\|\gamma(t/T)-[x_l(t),y_l(t)]^T\| \geq R$ for all $i  \in \lbrace 1,...,N \rbrace $ and all $t \in [0,T]$.
\end{definition}
Our objective is to derive conditions on the parameters $A$, $k$, and $\alpha$ and on the trajectory $\gamma(t)$ and its time derivative $\dot{\gamma}(t)$ which guarantee that system \ref{eq:MainSys} is $R$-confinement controllable about the trajectory $\gamma$.
The requirement that the leader must maintain a minimum distance from the trajectory differentiates our approach from other confinement strategies that use a single leader, which typically employ attraction-type interactions between the leader and followers.  Moreover, this distance requirement provides a way to prevent collisions between the leader and the followers without explicitly modeling the physical dimensions of the agents.

 %finite sizes.

%is not typically included in containment and confinement problems.  We introduce this requirement to 

%is non-standard as compared to standard containment or confinement problems. 
%However, this is mainly to emphasize that this is indeed possible using the method presented in this paper, unlike other methods where a single leader might be used for such a problem (for instance, using attraction-type interactions between followers and the leader). 

%Moreover, this also accounts for finite size of each follower agent in a swarm, which we do not incorporate explicitly in the model.

%\kar{the first paragraph here is not a part of this analysis. It is only a passing comment on the relevance of traditional geometric control theory with respect to our system. Maybe we could keep it in the previous section}.
System \ref{eq:MainSys} can be expressed in the affine control form $\dot{\mathbf{x}} = \mathbf{f}(\mathbf{x}) +\sum_i u_i \mathbf{g}_i(\mathbf{x})$, which enables the application of Lie algebraic conditions of geometric control to study its controllability properties \cite{sontag2013mathematical, sussmann1987general}. However, the system has no fixed point for any time-independent set of control inputs $\lbrace u_i \rbrace$. Therefore, one cannot conclude much more than accessibility at points of interest. While Lie algebraic conditions for controllability do exist for the case of non-rest conditions, they require the assumption of global bounds on the vector fields \cite{godhavn1999steering}, which is not a valid assumption for system \ref{eq:MainSys}. Alternatively, certain results on controllability about trajectories can be applied if the admissibility of the trajectories can be characterized beforehand \cite{bianchini1993controllability}. 

We take a different approach by using {\it averaging theory} to show that the followers can be stabilized about the target trajectory with a particular choice of control inputs.  Averaging theory is built on the principle that the behavior of a dynamical system with a rapidly oscillating vector field can be accurately represented by the averaged system behavior over a time interval that is dependent on a perturbation parameter $\epsilon$.  For small enough $\epsilon$, one can relate the solutions of the averaged system, a set of autonomous differential equations, to the solutions of the original non-autonomous differential equations.  Moreover, the stability of the solutions is easier to characterize in the averaged system than in the original system.  We apply averaging analysis to our system in the next section.

%derive the averaged system for our scenario and perform stability analysis in the next section.}

%the averaged behavior of the system is an accurate representation of the behavior of the original system on a time interval that is dependent on a parameter $\epsilon$.  

%of a system of non-autonomous differential equations is an accurate representation of the behavior of the original system, on a time interval that is dependent on a parameter $\epsilon$. 

% \kar{...due to the available sufficient conditions for controllability in literature?}

%\kar{Instead, in the following section, using averaging theory we show that for a particular choice of control inputs the followers can be stabilized about the required trajectory?}

% \spr{Why are we putting the system in standard averaging form now?} \kar{to be able to use the averaging result from Proposition \ref{Main}}.

% there has been some work on

%\kar{I was vague here on purpose. I am not really sure if it the conditions are very applicable, since the required computations (for proving controllability about trajectories) are tedious, and based on higher order necessary conditions of optimal control theory. But whatever I had written was right, and I think it's sufficient enough to let the reader know that we are aware of this work}

\section{Averaging Analysis}

In the forthcoming analysis, we prove that system \ref{eq:MainSys} is $R$-confinement controllable about a trajectory $\gamma(t)$ under suitable bounds on $\gamma(t)$ and $\dot{\gamma}(t)$. We also demonstrate that if a follower starts at the origin, then its distance from the trajectory can be made arbitrarily small while the leader maintains a distance $R$ from the trajectory.  To simplify the analysis, we consider the case where $R=1$ and $N=1$, i.e., there is a single follower. We show that the analysis also holds for any follower population $N$, since there are no interactions between the followers.  The results for general $R>0$ follow trivially from an appropriate scaling of the parameter $A$.  

% \kar{well, we don't exactly know which frequency does the job for us, but only that there exist some} 

% \spr{why is this important?}

% , the analysis also holds for any follower population $N$, as we show in \autoref{ConThe}

% (this is important to point out, because a simple attraction potential could also be used to achieve this otherwise).  

%We conduct the analysis for $R=1$

%Interestingly, it also results from the following analysis that, if a follower starts at the origin, then the distance between the follower and the trajectory can be made arbitrarily small 

%\spr{Also to derive the leader control inputs?} \kar{I am not sure if we are deriving the control inputs. But we have one and we show it works. This is why I mention in the next paragraph that we almost-constructively show this.} 

Denote the initial position of the leader by $[x_l(0),y_l(0)]^T =[1,0]^T$, and let $u_x(t) = -\omega \sin \omega t$ and $u_y(t) = -\omega \cos \omega t$ for some $\omega >0$. With these control inputs, one can explicitly solve for $x_l(t)$ and $y_l(t)$.  Then system \ref{eq:MainSys} can be expressed in the following reduced form,
\begin{eqnarray}
\label{eq:RedSys}
\dot{x} &=& \hspace{2mm} v_x  \\ \nonumber
\dot{v}_x &=& \frac{A (x-\cos \omega t)}{[(x-\cos \omega t)^2+(y-\sin \omega t))^2]^{\alpha}} - k v_x\\ \nonumber
\dot{y} &=& \hspace{2mm} v_y \\ \nonumber
\dot{v}_y &=& \frac{A (y-\sin \omega t)}{[(x-\cos \omega t)^2+(y-\sin \omega t))^2]^{\alpha}} - k v_y.
\label{eq:MainS2}
\end{eqnarray}

We change the time variable to $\tau = \omega t$ and define $\epsilon = 1/\omega$. Then the solution of \ref{eq:RedSys} satisfies,
\begin{eqnarray}
\frac{dx}{d\tau} &=& \epsilon v_x  \\ \nonumber
\frac{dv_x}{d\tau} &=&  \frac{\epsilon A (x-\cos \tau )}{[(x-\cos \tau )^2+(y-\sin \tau )^2]^{\alpha}} - \epsilon k v_x\\ \nonumber
\frac{dy}{d\tau} &=& \epsilon v_y \\ \nonumber
\frac{dv_x}{d\tau} &=&  \frac{\epsilon A (y-\sin \tau )}{[(x-\cos \tau )^2+(y-\sin \tau )^2]^{\alpha}} - \epsilon k v_y.
\label{eq:CoVRed}
\end{eqnarray}
Now the system can be expressed in the so-called standard averaging form,
\begin{equation}
\frac{d\mathbf{X}}{d\tau} = \epsilon \mathbf{f}(\tau,\mathbf{X}), 
\label{eq:TransRed}
\end{equation}
where $\mathbf{X}=[x, v_x, y, v_y]^T$ and
\begin{equation}
\mathbf{f}(\tau,\mathbf{X}) = 
\begin{bmatrix} 
\hspace{-2mm} &  v_x  \hspace{2mm} \\ 
\hspace{-2mm} & \frac{A (x-\cos \tau )}{[(x-\cos \tau )^2+(y-\sin \tau )^2]^{\alpha}} -  k v_x \hspace{2mm} \\
\hspace{-2mm} & v_y \hspace{2mm} \\
\hspace{-2mm} & \frac{A (y-\sin \tau )}{[(x-\cos \tau )^2+(y-\sin \tau )^2]^{\alpha}} - k v_y \hspace{2mm}
\end{bmatrix}.
\end{equation}
Then we can consider the {\it averaged system}, %\spr{what is the relationship between $x^r$ and $\bar{x}$, etc.?} \kar{(pointed out in \ref{ShortLemma})},
\begin{equation}
\dot{\bar{\mathbf{X}}} = \epsilon \mathbf{f}_{av}(\bar{\mathbf{X}}) \hspace{5mm} \bar{\mathbf{X}}(0) = \bar{\mathbf{X}}_0,
 \label{eq:AvgSys}
\end{equation}
where $\bar{\mathbf{X}} = [\bar{x},\bar{v}_x,\bar{y}_y,\bar{v}_y]^T$ and 
\begin{equation}
\mathbf{f}_{av}(\bar{\mathbf{X}}) = 
\begin{bmatrix}
\hspace{-2mm} & \bar{v}_x \hspace{1mm} \\
\hspace{-2mm} & \frac{1}{2\pi} \int^{2 \pi}_0 \frac{ A (\bar{x}-\cos (s) )}{[(\bar{x}-\cos (s) )^2+(\bar{y}-\sin (s) )^2]^{\alpha}} ds - k \bar{v}_x \hspace{1mm}\\
\hspace{-2mm} & \bar{v}_y \hspace{1mm} \\ 
\hspace{-2mm} &  \frac{1}{2\pi} \int^{2 \pi}_0 \frac{ A (\bar{y}-\sin (s) )}{[(\bar{x}-\cos (s) )^2+(\bar{y}-\sin (s) )^2]^{\alpha}} ds - k \bar{v}_y \hspace{1mm}
\end{bmatrix}.
\end{equation}

\vspace{2mm}

If the origin of the averaged system is an exponentially stable equilibrium point, then there exists a neighborhood of the origin for which solutions of the original system remain bounded and close to the origin for infinite time.  The following lemma gives conditions for exponential stability of the origin in the averaged system. The existence and uniqueness of solutions to the averaged system follow from this result for initial conditions sufficiently close to the origin.

\vspace{2mm}

\begin{lemma}
The origin is an exponentially stable equilibrium point of the averaged system \ref{eq:AvgSys} for all $k>0$ if $A>0$ and $\alpha > 1$.
\end{lemma}
\begin{proof}
The proof is based on the linearization principle.  The Jacobian map for $\mathbf{f}_{av}(\bar{\mathbf{X}})$ at $\bar{\mathbf{X}} = \mathbf{0}$ is given by (see the Appendix for computations):
\begin{equation}
D\mathbf{f}_{av}(\mathbf{0}) = 
\begin{bmatrix}
0 & 1 & 0 & 0 \\ 
-A(\alpha-1) & -k & 0 & 0 \\ 
0 & 0 & 0 & 1 \\ 
0 & 0 & -A(\alpha-1) & -k
\end{bmatrix}.
\label{eq:AvJacobian}
\end{equation}
The spectrum of $D\mathbf{f}_{av}(\mathbf{0})$ can be computed to be
\begin{equation}
\sigma(D\mathbf{f}_{av}(\mathbf{0})) =  -\frac{k}{2} \pm \frac{1}{2}(k^2 - 4A(\alpha-1))^{\frac{1}{2}}.
\label{eq:Avspec}
\end{equation}
Then the result follows from \cite{khalil2002nonlinear}[Corollary 4.3]. %\hspace{12mm}
\end{proof} 

\vspace{2mm}

The following proposition demonstrates that this result can be used to obtain estimates for the solutions of system \ref{eq:TransRed}.  The proposition gives bounds on solutions of general systems in which the vector field can be decomposed into (1) a highly oscillatory, state-dependent component, and (2) a state-independent component, which is not necessarily periodic, that evolves on a slower time scale.  By converting our system into this form through a coordinate transformation, we can use this result to construct control inputs for the leader to confine followers about a trajectory.  The proof is similar to the proofs of the averaging theorems in \cite{khalil2002nonlinear, guckenheimer1983nonlinear}.  Since our result does not immediately follow from these theorems, we adapt the initial steps of their proofs to our problem.

\begin{proposition}
\label{Main}
Let a function $\mathbf{f}(t,\mathbf{x},\epsilon)$ and its partial derivatives with respect to $(\mathbf{x},\epsilon)$ up to the first order be continuous and bounded for $(t,\mathbf{x},\epsilon) \in [0,\infty) \times U_0 \times [0,\epsilon_0]$ for every compact set $U_0 \subset U$, where $U \subset \mathbb{R}^n$, and some $\epsilon_0>0$. Additionally, let $\mathbf{g} : [0,\infty) \rightarrow \mathbb{R}^n$ be continuous. Suppose that $\mathbf{f}$ is $T$-periodic in $t$ for some $T>0$. Then $\mathbf{x} \equiv \mathbf{x}(t,\epsilon,\mathbf{g})$ and $\mathbf{z} \equiv \mathbf{z}(t)$ evolve according to the following equations:
\begin{equation}
\dot{\mathbf{x}} = \epsilon \mathbf{f}(t,\mathbf{x},\epsilon) + \epsilon \mathbf{g}(t), \hspace{5mm} \mathbf{x}(0,\epsilon,\mathbf{g}) = \mathbf{x}_0,
\label{eq:Gpedic}
\end{equation}
\begin{equation}
\dot{\mathbf{z}} = \epsilon \mathbf{f}_{av}(\mathbf{z}), \hspace{5mm} \mathbf{z}(0) = \mathbf{z}_0,
\label{eq:GenAvg}
\end{equation}
where $\mathbf{f}_{av}(\mathbf{z}) = \frac{1}{T} \int^T_0 \mathbf{f}(\tau,\mathbf{z},0) d \tau$.
Suppose that the origin $\mathbf{z}=\mathbf{0} \in U$ is an exponentially stable equilibrium point of the averaged system \autoref{eq:GenAvg}, $\Omega \subset U$ is a compact subset of its region of attraction, and $\mathbf{x}_0 \in \Omega$.  Then there exist positive constants $k$, $\delta^*$, $\epsilon^*$, $\beta_1$, $\beta_2$, and $\mu$  such that if $\|\mathbf{g}(t)\| \leq \delta$ for all $t \geq 0$, then the condition $\|\mathbf{x}(0,\epsilon,\mathbf{g}) - \mathbf{z}(0)\| \leq \mu$ implies that 
\begin{equation}
\|\mathbf{x}(t,\epsilon,\mathbf{g})-\mathbf{z}(t)\| ~\leq~ k e^{-\lambda \epsilon t} \|\mathbf{x}_0 - \mathbf{z}_0\| + \beta_1 \epsilon + \beta_2 \delta
\label{prop1Eq}
\end{equation}
for all $\epsilon \in (0,\epsilon^*)$, $\delta \in (0,\delta^*)$, and $t \geq 0$.
\end{proposition}
\begin{proof}
Define the functions
\begin{equation}
\mathbf{h}(t,\mathbf{x}) = \mathbf{f}(t,\mathbf{x},0) - \mathbf{f}_{av}(\mathbf{x})
\end{equation}
and
\begin{equation}
\mathbf{u}(t,\mathbf{x}) = \int^t_0 \mathbf{h}(\tau,\mathbf{x}) d\tau.
\end{equation}
Since $\mathbf{h}(t,\mathbf{x})$ is $T$-periodic in $t$ and has zero mean,  the function $\mathbf{u}(t,\mathbf{x})$ is $T$-periodic in $t$. Hence, $\mathbf{u}(t,\mathbf{x})$ is bounded for all $(t,\mathbf{x}) \in [0,\infty) \times \Omega$. Moreover, $\partial \mathbf{u} / \partial t$ and $D_{\mathbf{x}} \mathbf{u}$ are given by
\begin{eqnarray}
\frac{\partial \mathbf{u}}{\partial t} &=& \mathbf{h}(t,\mathbf{x}), \nonumber \\ 
D_{\mathbf{x}} \mathbf{u} &=& \int^t_0 D_{\mathbf{x}} \mathbf{h} (\tau,\mathbf{x}) d \tau.
\end{eqnarray}

Consider the change of variables
\begin{equation}
\mathbf{x}= \mathbf{y} + \epsilon \mathbf{u}(t,\mathbf{y}).
\label{eq:CoV}
\end{equation}
Differentiating both sides of \autoref{eq:CoV} with respect to $t$, we obtain
\begin{equation}
\dot{\mathbf{x}} = \dot{\mathbf{y}} + \epsilon \frac{\partial \mathbf{u}}{\partial t}(t,\mathbf{y}) + \epsilon D_{\mathbf{y}} \mathbf{u}(t,\mathbf{y}) \dot{\mathbf{y}}. 
\label{eq:CoVdxdt}
\end{equation}
Substituting \autoref{eq:Gpedic} for $\dot{\mathbf{x}}$ into \autoref{eq:CoVdxdt}, the new state variable $\mathbf{y}$ satisfies the equation
\begin{eqnarray}
\big [\mathbf{I} + \epsilon D_{\mathbf{y}} \mathbf{u} \big ] \dot{\mathbf{y}} &=& \epsilon \mathbf{f}(t,\mathbf{y} + \epsilon \mathbf{u},\epsilon) + \epsilon \mathbf{g}(t) - \epsilon \frac{\partial \mathbf{u}}{ \partial t}  \nonumber \\ 
&=& \epsilon \mathbf{f}(t,\mathbf{y} + \epsilon \mathbf{u},\epsilon) + \epsilon \mathbf{g}(t) - \epsilon \mathbf{f}(t,\mathbf{y},0) \nonumber \\ 
&& +  \epsilon  \mathbf{f}_{av}(\mathbf{y}). \nonumber
\end{eqnarray}
The assumption of bounds on the partial derivatives of $\mathbf{f}$ over $(t,\mathbf{x}) \in [0, \infty) \times U_0$ imply that for small enough $\epsilon$, $\big [\mathbf{I} + \epsilon D_{\mathbf{y}} \mathbf{u} \big ]$ is invertible.  Hence, the state variable $\mathbf{y}$ satisfies
\begin{eqnarray}
\dot{\mathbf{y}} &=& \epsilon  \mathbf{f}_{av}(\mathbf{y}) +\epsilon \mathbf{g}(t) + \epsilon \big[\epsilon D_{\mathbf{y}}\mathbf{f}(t,\mathbf{y},0)\mathbf{u}+ \epsilon \frac {\partial \mathbf{f}}{\partial \epsilon}(t,\mathbf{y},0) \nonumber \\ 
& \hspace{5mm}& -  \epsilon D_{\mathbf{y}}\mathbf{u}\mathbf{f}_{av}(\mathbf{y}) -  \epsilon D_{\mathbf{y}}\mathbf{u}\mathbf{g}(t) \big] + O(\epsilon ^3) \nonumber \\ 
&\vcentcolon =& \epsilon  \mathbf{f}_{av}(\mathbf{y}) +\epsilon \mathbf{g}(t) + \epsilon^2 \mathbf{q}(t,\mathbf{y},\epsilon,\mathbf{g}), \label{eq:dydtOrig} %\nonumber
\end{eqnarray}
where $\mathbf{q}(t,\mathbf{y},\epsilon,\mathbf{g})$ is not necessarily periodic. From the assumptions on $\mathbf{f}(t,\mathbf{x},\epsilon)$ and $\mathbf{g}(t)$, $\mathbf{q}(t,\mathbf{y},\epsilon,\mathbf{g})$ is continuous and bounded for $(t,\mathbf{x},\epsilon) \in [0,\infty) \times \Omega \times [0,\epsilon_0]$.

Now consider the change of variables $s = \epsilon t$. Then the solutions of \autoref{eq:dydtOrig} and \autoref{eq:GenAvg} satisfy
\begin{equation}
\frac{d\mathbf{y}}{ds} = \mathbf{f}_{av}(\mathbf{y}) + \mathbf{g}(s) +\epsilon \mathbf{q}(s/\epsilon,\mathbf{y}, \epsilon,\mathbf{g})
\end{equation}
and 
\begin{equation}
\frac{d\mathbf{z}}{ds} = \mathbf{f}_{av}(\mathbf{z}), \label{eq:dzds}
\end{equation}
respectively.  By assumption, the origin $\mathbf{z} = \mathbf{0}$ is an exponentially stable equilibrium point of \autoref{eq:dzds}. Then by \cite{khalil2002nonlinear}[Theorem 9.1], there exist positive constants $\breve{\epsilon}^*$, $\delta^*$, $\mu$, $\breve{\beta}_1$, $\beta_2$, $\delta$, $k$, and $\lambda$ such that the conditions $\|\mathbf{g}(s)\| < \delta$ and $\|\mathbf{y}_0 - \mathbf{z}_0\| \leq \mu$ imply that,  
\begin{equation}
\|\mathbf{y}(s)-\mathbf{z}(s)\| ~\leq~ k e^{- \lambda s} \|\mathbf{y}_0 - \mathbf{z}_0\| + \breve{\beta}_1 \epsilon + \beta_2 c  
\label{eq:Pestimate}
\end{equation}
for all $\epsilon \in (0,\breve{\epsilon}^*)$, $\delta \in (0,\delta^*)$, and $s \geq 0$.  For small enough $\epsilon$, the transformation \ref{eq:CoV} is a one-to-one mapping from $\Omega$ to its image for all $t \geq 0$ and $\epsilon \in [0,\epsilon_{o}]$ for some $\epsilon_{o} > 0$ \cite{sanders2007averaging}[Lemma 2.8.3]. Therefore, $\mathbf{x}-\mathbf{y} = O(\epsilon)$. Hence, it follows from the triangle inequality and the estimate \ref{eq:Pestimate} that there exist positive constants $\delta^*$, $\epsilon^*$, $\mu$, and  $c$  such that the conditions $\|\mathbf{g}(t)\| \leq \delta$ and $\|\mathbf{x}(0,\epsilon,\mathbf{g}) - \mathbf{z}(0)\| \leq \mu$ imply that \autoref{prop1Eq} is true for all $\epsilon \in (0,\epsilon^*)$, $\delta \in (0,\delta^*)$, and $t \geq 0$. %  \hspace{35mm}
\end{proof}

\vspace{2mm}

Note that this result is not the same as general averaging, which is a generalization of the above proposition for the case where $\mathbf{g}= \mathbf{0}$ and $\mathbf{f}(t,\mathbf{x},\epsilon)$ is required to be only almost-periodic in $t$ \cite{khalil2002nonlinear, sanders2007averaging}, since we have explicit bounds on the solution of \autoref{eq:GenAvg} that depend on $\epsilon$ as well as $\delta$. One can impose the additional assumption that $\mathbf{f}(t,\mathbf{x},\epsilon)$ has partial derivatives up to the second order. Under this assumption, $\mathbf{g}=\mathbf{0}$ implies the existence of exponentially stable periodic orbits that approach the equilibrium point of the averaged system as $\epsilon$ tends to $0$. See, for example, \cite{guckenheimer1983nonlinear}. In this case, one can interpret the result as ``uniform'' input-to-state stability of the system about $\epsilon$-dependent exponentially stable $T$-periodic orbits. Our result can also be seen as a special case of the results in \cite{nevsic2001input}, where input-to-state stability of general time-varying systems via averaging is analyzed using the notion of strong and weak averages.

In the next two results, we use Proposition \ref{Main} to show the existence of leader control inputs that can locally stabilize the followers about general trajectories in the plane that start at the origin. We consider leader inputs with an oscillatory component, as in \autoref{eq:MainS2}, and an additional component that depends on the target trajectory to be tracked.  Using a coordinate transformation, we show that this system can be converted to the form in \autoref{eq:GenAvg}, and hence we can obtain bounds on the distances between the follower positions and the target trajectory. % to be tracked.

\noindent{\it Remark on notation:}  We say that a function $\mathbf{f}:U \rightarrow \mathbb{R}$, where $U \subset \mathbb{R}^n$, is in $C^1(U)$ if it has continuous derivatives up to order $1$. 

\vspace{2mm}

\begin{lemma}
Suppose that $\gamma_1$ and $\gamma_2$ are elements in $C^1([0,1])$ such that $\gamma_1(0) = \gamma_2(0)=0$, and define the trajectory $\mathbf{\gamma}(t) = (\gamma_1(t),\gamma_2(t))$. Additionally, let $\gamma_1^\delta$ and  $\gamma_2^\delta$ be elements in $C^1([0,1/ \delta])$ such that $\gamma^\delta_1(t) = \gamma_1(\delta t)$ and $\gamma^\delta_2(t) = \gamma_2(\delta t)$ for each $t \in [0,1/\delta]$, where $\delta>0$.  Define $\mathbf{\gamma}^\delta(t) = (\gamma_1^\delta(t),\gamma^\delta_2(t))$.  Let $\mathbf{X}=[x, v_x, y, v_y]^T$ and define $\mathbf{F}(t,\mathbf{X},\omega,\mathbf{\gamma}^\delta)$ as
\begin{equation}
\mathbf{F}(t,\mathbf{X},\omega,\mathbf{\gamma}^\delta) = 
\begin{bmatrix}
\hspace{-2mm} & v_x \hspace{1mm} \\
\hspace{-2mm} & \frac{A (x-\cos \omega t - \gamma^\delta_1(t))}{[(x-\cos \omega t- \gamma^\delta_1(t) )^2+(y-\sin \omega t - \gamma^\delta_2(t))^2]^{\alpha}} -  k v_x \hspace{1mm} \\
\hspace{-2mm} & v_y \hspace{1mm} \\
\hspace{-2mm} &  \frac{A (y-\sin \omega t- \gamma^\delta_2(t) )}{[(x-\cos \omega t- \gamma^\delta_1(t) )^2+(y-\sin \omega t - \gamma^\delta_2(t))^2]^{\alpha}} - k v_y \hspace{1mm} \\
\end{bmatrix}
\label{eq:Fdefinition}
\end{equation}
where $A>0$, $k>0$, and $\alpha>1$. Then consider the system 
\begin{equation}
\dot{\mathbf{X}} = \mathbf{F}(t,\mathbf{X},\omega,\mathbf{\gamma}^\delta), \hspace{5mm} \mathbf{X}(0) = \mathbf{X}_0.
\label{eq:PreEqX}
\end{equation}
There exist positive constants $r_1>0$, $\omega_1 > 0$, and $\delta_1 > 0$, which depend on $A$, $k$, and $\alpha$, such that $\omega >\omega_1$, $\delta >\delta_1$ and $\|\mathbf{X}_0\| \leq  r_1$ implies that the solution $\mathbf{X}(t)$ to \autoref{eq:PreEqX} is defined over the interval $t \in [0,1/\delta]$ and that
\begin{equation}
\|(x(t),y(t) - (\gamma^\delta_1(t),\gamma^\delta_2(t)) \| < 1   
\label{eq:estimate_1}
\end{equation}
for all $t \in [0,1/\delta]$.
Moreover, if $\mathbf{X}_0 = \mathbf{0}$,  then for each $\nu > 0$ there exist $\omega_2 > 0$ and $\delta_2 > 0$ such that for all $\omega > \omega_2$ and  $\delta \in (0,\delta_2)$,
\begin{equation}
\|(x(t),y(t) - (\gamma^\delta_1(t),\gamma^\delta_2(t)) \| < \nu
\label{eq:estimate_nu}
\end{equation}
for all $t \in [0,1/\delta]$.
\label{ShortLemma}
\end{lemma}
\begin{proof}
Consider the change of variables $\mathbf{E} = \mathbf{X} - [\gamma^\delta_1(t),0,\gamma^\delta_2(t),0]^T$. Then we have the initial value problem 
\begin{equation}
\dot{\mathbf{E}} = \mathbf{F}(t,\mathbf{E},\omega,\mathbf{0}) + \delta \mathbf{g}(t), \hspace{5mm} \mathbf{E}(0) = \mathbf{X}_0,
\label{eq:ErSys}
\end{equation}
where $\mathbf{g}(t) = [\dot{\gamma}_1(\delta t),0,\dot{\gamma}_2(\delta t),0]^T$. Let $\epsilon = 1/\omega$ and consider the change of variable $s = t/\epsilon$. Then the solution of \autoref{eq:ErSys} satisfies
\begin{equation}
\frac{d\mathbf{E}}{ds} = \epsilon \mathbf{F}(s,\mathbf{E}(s),1,\mathbf{0}) + \epsilon \delta \mathbf{g}(\epsilon s).
\end{equation}
The estimates in \autoref{eq:estimate_1} and \autoref{eq:estimate_nu} follow from Proposition \ref{Main}. Note that the dependence of $\mathbf{g}$ on the $\epsilon s$ time scale does not affect the proof, since $\mathbf{g}$ is still bounded from above as required.  %\hspace{66mm}
\end{proof}

\vspace{2mm}

For the case where the follower agent starts at the origin, this proposition implies that one can choose leader control inputs that drive the follower to track the trajectory $\gamma$ arbitrarily closely, albeit over a longer period of time than if the required bounds on the tracking error were weaker. This is a stronger result than is needed to ensure general $R$-confinement controllability, for which it is sufficient for the follower agent to be within a distance $R$ of the trajectory.

\vspace{2mm}

\begin{theorem}                                                                                                                                                                                                                                                                                                                                                                                  
\label{ConThe}
Suppose that $[x_l(0),y_l(0)]$ = $[1,0]^T$. Then there exists $r>0$ such that  $\|[x^i,y^i,v^i_x,v^i_y]\| < r$ for all $i = 1,2....N$ implies that system \ref{eq:MainSys} is 1-confinement controllable about all trajectories $\gamma$ for which $\mathbf{\gamma} \in \lbrace (\gamma_1,\gamma_2) : \gamma_1,\gamma_2 \in C^1([0,1]),\hspace{2mm} \gamma_1(0) = \gamma_2(0) = 0 \rbrace$.
\end{theorem}
\begin{proof}
Let $\omega$ and $\delta$ be positive scalars. Consider the following control inputs for the leader:
\begin{eqnarray}
u_x(t) &=& -\omega \sin \omega t + \delta \dot{\gamma}_x(\delta t) \nonumber \\ 
u_y(t) &=& \omega \cos \omega t + \delta \dot{\gamma}_y(\delta t).  \label{eq:MainContr}
\end{eqnarray}
Integrating these functions from $0$ to $t$ for each $t \in$ $[0,1/\delta]$, we can explicitly solve for the state variables $x_l(t)$ and $y_l(t)$: 
\begin{eqnarray}
 x_l(t) &=&  \cos (\omega t)+\gamma_x(\delta t), \nonumber \\
 y_l(t) &=& \sin (\omega t)+\gamma_x(\delta t). \label{eq:xlylSolns}
 \end{eqnarray}
Define $\mathbf{F}$ as in \autoref{eq:Fdefinition} and substitute \autoref{eq:xlylSolns} into system \ref{eq:MainSys}.  Then the state variables $x^i, v^i_x, y^i, v^i_y$ in the solution of system \ref{eq:MainSys} satisfy %to the system of equations
\begin{equation}
\frac{d \mathbf{X}^i}{dt} = \mathbf{F}(t,\mathbf{X}^i,\omega,\gamma^\delta), \hspace{5mm} \mathbf{X}^i(0) = {\mathbf{X}}^i_0,
\end{equation}
for each $i \in \lbrace 1,2,...,N \rbrace$, where $\mathbf{X}^i = [x^i,v^i_x,y^i,v^i_y]^T$.  The result follows from Lemma \ref{ShortLemma}. %\hspace{43mm} 
\end{proof}

\vspace{2mm}

The assumption that $[x_l(0),y_l(0)]$ = $[1,0]^T$ is not too restrictive. If $[x_l(0),y_l(0)]^T$ is any element on the unit circle, then the leader control inputs are  
\begin{eqnarray}
\label{eq:phacon} 
u_x(t) &=& -\omega \sin (\omega t + \phi) + \delta \dot{\gamma}_x(\delta t) \nonumber \\ 
u_y(t) &=& \omega \cos (\omega t + \phi) + \delta \dot{\gamma}_y(\delta t) 
\end{eqnarray}
for an appropriately chosen $\phi$.  These inputs will have the same effect on the system.  Due to the continuous dependence of the solutions of system \ref{eq:MainSys} on its initial conditions and parameters, it can be argued further that there exists a neighborhood of the unit circle from which the leader can always be driven to the circle, after which it can execute the control inputs defined in \autoref{eq:phacon}.

%one can always drive the leader to the unit circle and then execute the control inputs as defined in \autoref{eq:phacon}. \
%\kar{
\begin{remark}
We can also consider a more general setting in which the follower agents are running a consensus protocol while being confined by the leader. In this case, the followers' interactions can be modeled using an undirected simple graph with fixed topology, $\mathcal{G} = (\mathcal{V}, \mathcal{E})$, where $\mathcal{V}$ is the set of follower agents and $\mathcal{E}$ identifies the pairs of agents that interact with each other.  The graph is not necessarily connected. Then the $x$-coordinates of the follower agents' positions and velocities evolve according to 
\begin{eqnarray}
\label{eq:MainSys}
\dot{x}^i &=& v^i_x  \\ \nonumber
\dot{v}^i_x &=& \frac{A (x^i-x_l)}{[(x^i-x_l)^2+(y^i-y_l)^2]^{\alpha}} - k v^i_x  \\ \nonumber &\hspace{2mm}& - \sum_{j \in \mathcal{N}(i)}(x^i-x^j) - \sum_{j \in \mathcal{N}(i)}(v_x^i-v_x^j)~, 
\end{eqnarray}
with similar dynamics for the $y$-coordinates.  As in \autoref{eq:AvJacobian}, the Jacobian map for the averaged vector field $\mathbf{f}_{av,c}$ of the composite group of followers can be computed as 
\begin{equation}
D\mathbf{f}_{av,c}(\mathbf{0}) = 
\begin{bmatrix}
\mathbf{B} & 0 \\ 
0 & \mathbf{B} \\ 
\end{bmatrix},
\end{equation}
where
\begin{equation}
\mathbf{B} = 
\begin{bmatrix}
0 & \mathbf{I}_N \\ 
- \beta \mathbf{I}_N-\mathbf{\mathcal{L}} & -k\mathbf{I}_N - \gamma \mathbf{\mathcal{L}} \\ 
\end{bmatrix},
\end{equation}
in which $\mathbf{I}_N$ is the identity matrix of size $N \times N$, $\mathbf{\mathcal{L}}$ is the Laplacian of the graph $\mathcal{G}$, and $\beta = A(\alpha -1)$, $\gamma$, and $k$  are positive scalars.  The eigenvalues of $D\mathbf{f}_{av,c}(\mathbf{0})$  can be computed according to \cite{ren2008distributed}[Chapter 4] as:
\begin{equation}
\lambda_{i\pm} = -\frac{k}{2}+ \frac{ \gamma \mu_i \pm \sqrt{(\gamma \mu_i-k)^2-4(\beta - \mu_i)}}{2}~,
\end{equation}
where $\mu_i$ is the $i^{th}$ eigenvalue of $-\mathcal{L}$. Since the spectrum of $D\mathbf{f}_{av,c}(\mathbf{0})$ lies in the open left half-plane of $\mathbb{C}$ whenever $\beta = A(\alpha-1), ~k,$ and $\gamma$ are positive, this implies that confinement controllability of the leader-follower system is retained under inter-follower interactions that can be modeled as a graph with fixed topology. 
\end{remark} % $\lambda_{i\pm}$ are the eigenvalues associated with the eigenvalue 
%}

\section{Simulations} \label{sec:Simulations}

%\subsection{Simulation Setup}

We validate our confinement control approach with simulations of two scenarios with different parameter values.  In Case 1, the target trajectory is $\gamma(\theta) = [2 \cos 2 \pi \theta \hspace{2mm} 2 \sin 2 \pi \theta]^T $, the final time condition is $T = 40 \pi$, and the other parameters are $A = 1$, $k = 0.05$, $\alpha = 10$, and $\omega =10$.  In Case 2, we set $\gamma(\theta) = [10 \cos 4 \pi \theta \hspace{2mm} 2 \pi \theta]^T $, $T= 80 \pi$, $A = 1$, $k = 0.03$, $\alpha = 5$, and $\omega =20$.  In both cases, the leader control inputs are chosen according to \autoref{eq:MainContr}, and there are three follower agents.  The initial positions and velocities of the followers are drawn uniformly randomly from a ball of radius  0.125 (Case 1) or 0.5 (Case 2) that is centered at the origin.  
%The initial conditions of the followers for Case $1$ and Case $2$ were taken to be in a ball of radius 0.125 and $0.5$ respectively, around the origin, using a random number generator.  % The final time condition is

%\spr{What is the effect of $\alpha$ and $\omega$?  Or why did we choose the values that we did?}. 

%The difference in $\omega$ doesn't really show in these plots though, since both seem to be eventually converging to the trajectory. This is because the neighborhood for the first case is taken to be 0.5 instead of 0.125. 
%\subsection{Results}

\autoref{fig:Sims} shows the agent trajectories during the two simulated scenarios. In both cases, the followers remain well within a distance of 1 from the target trajectory, which confirms that the system is $1$-confinement controllable about this trajectory.  The followers in Case 2 display larger initial oscillations about the target trajectory than the followers in Case 1.  This is because the followers in Case 2 have a lower damping coefficient $k$ and initial velocities that were chosen from a larger neighborhood of the origin, and thus tended to be higher.  
%\spr{In addition, the higher value of $\alpha$ in Case 1 causes the followers to experience a larger repulsion force from the leader, driving them closer to the target trajectory?}  \kar{I am not sure. There are cases where if $\alpha$ is too high and leader and follower are close enough, the follower rockets out. This can ofcourse b}

We also note that the higher value of $\omega$ in Case 2 implies that the region of attraction of the original system is closer to the region of the attraction of the averaged system than in Case 1.  This caused the confinement strategy in Case 2 to be successful for a larger ball of possible initial conditions than in Case 1.  Indeed, we found that the followers in Case 1 did not remain confined when the initial conditions were drawn from outside the ball of radius 0.125.

%The followers did not remain confined for the 1st case when the initial conditions were taken outside the $0.125$ ball due to the lower frequency

%\kar{A higher value of $\alpha$ implies the follower feels larger repulsion force when the distance between the leader and follower is less than $1$, as compared to lower values of $\alpha$.}
%\kar{ }

\begin{figure}
 
\centering
    \subfigure[Case $1$]
    {
     \label{fig:SimsCase1}
   \includegraphics[trim = 60mm 50mm 60mm 57mm, scale=0.4]{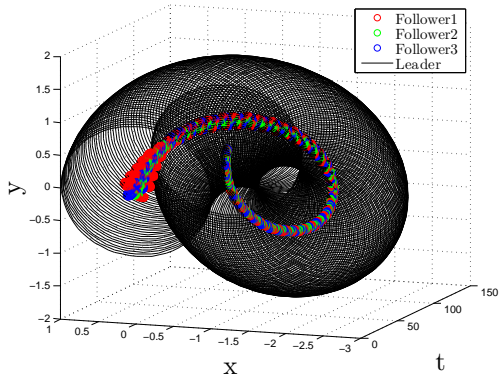}
      }
   
      \subfigure[Case $2$]
    {
     \label{fig:SimsCase2}
   \includegraphics[trim = 60mm 50mm 60mm 60mm, scale=0.4]{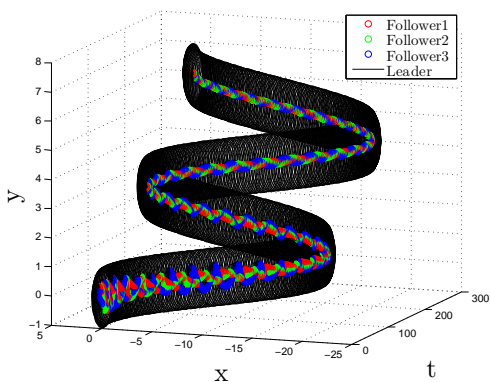} }
      \caption{Evolution of leader and follower agent positions over time.  The followers start at random positions within a neighborhood of the origin.} %  \spr{around the target trajectories?}
  \label{fig:Sims}   
   \end{figure}

\section{Robot Experiments}

\begin{figure}[t]
\centering
\vspace{2mm}
\includegraphics[trim = 0cm 0mm 0cm 0mm, clip, width=6cm]{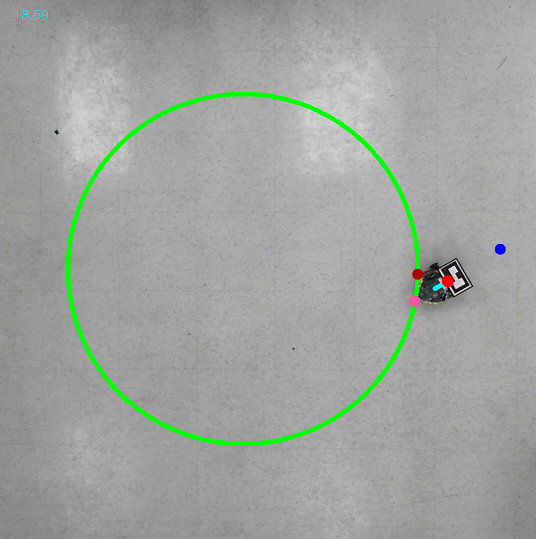} 
\caption{The robot being herded by a virtual leader along a target trajectory $\gamma_E$, shown in green.  The blue dot is the virtual leader's position, the pink dot is the center of the leader's orbit, and the dark red dot is the position of a simulated follower agent.  The light blue line on the robot indicates its orientation.}  %holonomic theory simulation location.
\label{fig:circleCameraView}
\end{figure}

%The experimental test bed as seen from the overhead camera.  The
% The robot being herded along trajectory $\gamma_E$, the green line. 

\subsection{Experimental Setup}

We also tested our confinement control approach with an experiment in which a small differential-drive robotic platform {\it Pheeno} \cite{WilsonRAL2015} emulated a follower agent that was confined by a virtual leader. The experiment was performed in the ASU Autonomous Collective Systems Laboratory. One robot was placed in a $1.5$ m $\times$ $2.1$ m arena and marked with a 2D binary identification tag to allow for real-time position and orientation tracking. The tags were tracked using one overhead Microsoft LifeCam Studio Webcam with a resolution of $1920 \times 1080$ pixels at a frame rate of 30 FPS. The 2D binary tags were identified using standard OpenCV algorithms on a Windows computer.

%To accommodate the nonholonomic constraint of the robot platform, a standard PI controller was used to orient the robot's heading away from the virtual leader at all times, producing approximate holonomic motion of the robot along the target trajectory.  This motion was enabled by setting the virtual leader's oscillation frequency low enough for the robot to reliably control its heading during the experiment. A desktop computer updated and transmitted the virtual leader's position and the robot's global position to the robot via WiFi at a frequency of 3 Hz.  

To accommodate the nonholonomic constraint of the robot platform, the robot used a PI controller to keep its heading facing away from the virtual leader at all times. This heading control produces approximate holonomic motion in a differential-drive platform. To facilitate this motion, the leader's oscillation frequency was maintained below the bandwidth of the robots' motors. A desktop computer updated and transmitted the virtual leader's position and the robot's global position to the robot via WiFi at a frequency of 3 Hz. 

% For this approximation to work

% kept low enough to remain under 

%the robot used a PI controller to keep its heading facing away from the leader at all times. 

%This allows for approximate holonomic motion in a differential drive platform. 

%For this approximation to work, the virtual leader's oscillation frequency was kept low enough for the robot to reliably control its heading during the experiment. 

%the confinement strategy was run on a control computer to constrain a waypoint.  
%This waypoint was updated using the algorithm described above.
%The robot then used a standard PID controller and onboard odometry to ``chase'' the waypoint.

We defined a circular target trajectory, shown in \autoref{fig:circleCameraView}, for the robot to track in the experiment. The trajectory was given by $\gamma_E(t) = [R_{p}\cos(\omega_{p} t) \hspace{2mm} R_{p}(\sin \omega_{p} t)]^T$, where $R_{p} = 67$ cm and $\omega_{p} = 0.02$.  The robot was controlled using the parameters $A = R_{L}^{2\alpha -1}$, $R_{L} = 38$ cm, $k = 1$, $\alpha = 4$, and $\omega = 2$.

\begin{figure}[t]
\centering
\vspace{2mm}
\includegraphics[trim = 4mm 0mm 0cm 0mm, clip, scale=0.33]{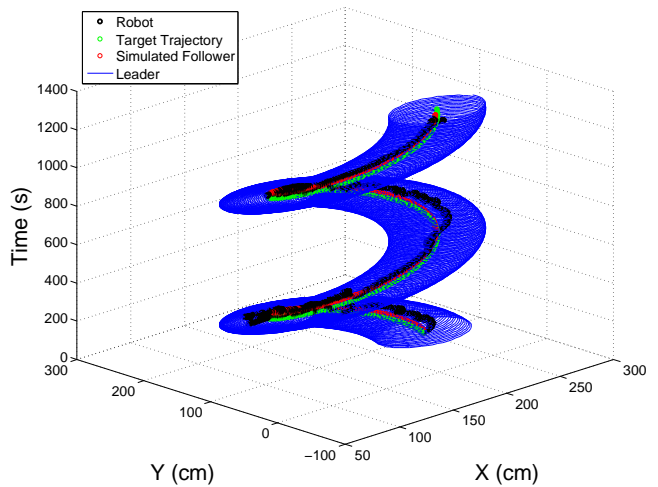} 
\caption{Time evolution of the positions of the leader agent (blue), simulated follower agent (red), and robot (black) for an experiment with the target trajectory $\gamma_E$ (green).}
\label{fig:experimentalResultsCircle}
\end{figure}

\subsection{Results}

We compared the trajectory of the robot during the experiment with the corresponding trajectory of a simulated follower agent, indicated by a dark red dot in \autoref{fig:circleCameraView}, that moves under the influence of the same leader agent.  Note that while the robot is nonholonomic, the simulated follower is holonomic.  \autoref{fig:experimentalResultsCircle} plots the trajectories of the virtual leader, the robot, and the simulated follower agent during the experiment.  The robot was initially placed inside the region of confinement bounded by the leader's path, and the simulated follower was initialized at the center of the leader's orbit. The figure illustrates that the robot was successfully kept inside the region of confinement for the duration of the experiment.  

\autoref{fig:experimentalResultsCircle} shows that the trajectory of the robot deviates farther from $\gamma_E$ than the trajectory of the simulated follower.  This is likely due to the nonlinear effects of friction on the shafts of the robots' motors.  The motors require a minimum input voltage to drive the robot's motion, which can be controlled at speeds above 2 cm/s.  These speeds are generated by the repulsive interaction potential in the robot's controller when the distance between the robot and the leader is below a threshold value.  Hence, the robot moves away from the leader only when it is within this distance; otherwise, it is stationary.  In contrast, the simulated follower does not have such a constraint on its speed and will move when the leader is farther away, causing its trajectory to adhere more closely to $\gamma_E$.  Thus, the experiment shows that when our confinement strategy is implemented in practice, the controller design must account for unmodeled dynamics in the physical platform by allowing only those velocity control inputs that exceed an appropriate lower bound.

\section{Conclusion}

We presented a leader-follower control strategy for confining a group of agents to a region around a target trajectory.  By constructing oscillatory leader inputs, we used averaging theory to show that follower agents can be confined to a neighborhood of the origin. 
From a control design perspective, one could investigate the quantitative effects of different parameters on the stability properties of the system. 
Well-developed computational tools can be used to estimate the size of the region of attraction of the origin \cite{chesi2011domain, kamyar2014polynomial}.  One could also augment this open-loop strategy with feedback to develop a more robust stabilization mechanism.  Incorporating multiple leaders and inter-agent collision avoidance into the control strategy is another direction for future work.
%, \kar{along with collision avoidance capability}

%\st{In addition, we showed that an input-to-state stability type argument for averaged systems can be used to design control inputs for reference tracking using a simple coordinate transformation.  This result has potential applications in other scenarios where a system can be stabilized using oscillatory control inputs.  Hence, our control strategy can be extended to more general trajectory tracking problems.}  

%\addtolength{\textheight}{-6cm}   % This command serves to balance the column lengths

%%%%%%%%%%%%%%%%%%%%%%%%%%%%%%%%%%%%%%%%%%%%%%%%%%%%%%%%%%%%%%%%%%%%%%%%%%%%%%%%

%%%%%%%%%%%%%%%%%%%%%%%%%%%%%%%%%%%%%%%%%%%%%%%%%%%%%%%%%%%%%%%%%%%%%%%%%%%%%%%%

\vspace{2mm} 

\section{Acknowledgment}
K.E. thanks Matthias Kawski for useful suggestions regarding the proofs.

\vspace{2mm} 

%%%%%%%%%%%%%%%%%%%%%%%%%%%%%%%%%%%%%%%%%%%%%%%%%%%%%%%%%%%%%%%%%%%%%%%%%%%%%%%%
\section*{APPENDIX}

%\paragraph{Notation} 
Let $U$ be a subset of $\mathbb{R}^n$.  The Jacobian map for a function $\mathbf{f}:U \rightarrow \mathbb{R}^m$ at a state $\mathbf{x} \in U$ is given by:
\begin{equation}
D\mathbf{f}(\mathbf{x}) = 
\begin{bmatrix}
\frac{\partial f_1}{\partial x_1}(\mathbf{x}) & \cdots & \frac{\partial f_1}{\partial x_n}(\mathbf{x}) \\
\vdots & \ddots & \vdots \\
\frac{\partial f_m}{\partial x_1}(\mathbf{x}) & \cdots & \frac{\partial f_m}{\partial x_n}(\mathbf{x}) \\
\end{bmatrix}. \nonumber
\end{equation}

Let $U$ be a subset of the unit sphere in $\mathbb{R}^4$, and define the function $\mathbf{f}:U \rightarrow \mathbb{R}^4$ as:
\begin{equation}
\mathbf{f}(\mathbf{x}) = 
\begin{bmatrix}
\hspace{-2mm} & x_2 \hspace{2mm} \\
\hspace{-2mm} & \frac{1}{2\pi} \int^{2 \pi}_0 \frac{ A (x_1-\cos (s) )}{[(x_1-\cos (s) )^2+(x_3-\sin (s) )^2]^{\alpha}} ds - k x_2 \hspace{2mm}\\
\hspace{-2mm} & x_4 \hspace{2mm}\\ 
\hspace{-2mm} &  \frac{1}{2\pi} \int^{2 \pi}_0 \frac{ A (x_3-\sin (s) )}{[(x_1-\cos (s) )^2+(x_3-\sin (s) )^2]^{\alpha}} ds - k x_4 \hspace{2mm}
\end{bmatrix}, \nonumber
\end{equation}
where $A,k,\alpha \in \mathbb{R}^+$.  Define $g_1(x,t) = x-\cos(t)$, $g_2(x,t)=x-\sin(t)$, and $h(x,y,t) = g_1(x,t)^2 +g_2(y,t)^2$. The function $\mathbf{f}$ is continuously differentiable everywhere in $U$ and has a well-defined Jacobian at the origin. We compute the nontrivial partial derivatives of the Jacobian map $D\mathbf{f}(\mathbf{x})$ using Leibniz's rule for partial derivatives of integrals:
\begin{align}
\frac{\partial f_2}{\partial x_1} =& \frac{A}{2\pi} \int^{2 \pi}_0 \frac{ h(x_1,x_3,s)^{\alpha}-2 \alpha g_1(x_1,s)^2h(x_1,x_3,s)^{\alpha-1}}{h(x_1,x_3,s)^{2 \alpha}} ds \nonumber \\ 
\frac{\partial f_2}{\partial x_3} =& \frac{A}{2\pi} \int^{2 \pi}_0 \frac{- 2 \alpha g_1(x_1,s) g_2(x_3,s) h(x_1,x_3,s)^{\alpha-1}}{h(x_1,x_3,s)^{2 \alpha}} ds \nonumber \\ 
\frac{\partial f_4}{\partial x_1} =& \frac{A}{2\pi} \int^{2 \pi}_0 \frac{- 2 \alpha g_1(x_1,s)g_2(x_3,s)h(x_1,x_3,s)^{\alpha-1}}{h(x_1,x_3,s)^{2 \alpha}} ds \nonumber \\ 
\frac{\partial f_4}{\partial x_3} =& \frac{A}{2\pi} \int^{2 \pi}_0 \frac{ h(x_1,x_3,s)^{\alpha} - 2 \alpha g_2(x_3,s)^2h(x_1,x_3,s)^{\alpha-1}}{h(x_1,x_3,s)^{2 \alpha}} ds. \nonumber
\end{align}

Evaluating these partial derivatives at the origin, we obtain:
\begin{eqnarray}
\frac{\partial f_2}{\partial x_1}(\mathbf{0}) &=& \frac{A}{2\pi} \int^{2 \pi}_0  (1 -2 \alpha \sin^2(s))ds ~=~ -A(\alpha-1)  \nonumber \\
\frac{\partial f_2}{\partial x_3}(\mathbf{0}) &=& \frac{A}{2\pi} \int^{2 \pi}_0 - 2 \alpha \sin(s)\cos(s)ds ~=~ 0  \nonumber \\
\frac{\partial f_4}{\partial x_1}(\mathbf{0}) &=& \frac{A}{2\pi} \int^{2 \pi}_0 - 2 \alpha \sin(s)\cos(s) ds ~=~ 0 \nonumber \\
\frac{\partial f_4}{\partial x_3}(\mathbf{0}) &=& \frac{A}{2\pi} \int^{2 \pi}_0 (1 - 2 \alpha \cos^2(s)) ds ~=~ -A(\alpha-1). \nonumber
\end{eqnarray}

\vspace{2mm}

%%%%%%%%%%%%%%%%%%%%%%%%%%%%%%%%%%%%%%%%%%%%%%%%%%%%%%%%%%%%%%%%%%%%%%%%%%%%%%%%

\bibliographystyle{plain}        % Include this if you use bibtex 

%\newpage
\bibliography{accref}  

\end{document}